%% file: main.tex
\documentclass[11pt]{article}
\usepackage[dvipsnames]{xcolor}
\usepackage{mathpazo}
\usepackage{fullpage}
\usepackage{enumitem}
\usepackage{pras}

\usepackage{booktabs} 
\usepackage{array}  

\usepackage[showdeletions]{color-edits}
\addauthor{pras}{violet}
\addauthor{jh}{Green}

\newtheorem{Alg}{Algorithm}
\newcommand{\myalg}[4][0cm]{
\medskip
\small{
\fbox{
\parbox{5.2in}{\vspace{#1}
\begin{Alg}\label{#2}{\textsl{ #3}}
\vspace{.1cm}\\ \emph{ #4}
\end{Alg}
}}
\medskip
}}

\newcommand{\cost}{\mathrm{cost}}
\newcommand{\ctopk}{\mathrm{cost}_{k}}

\newcommand{\distortion}{\mathrm{distortion}}
\newcommand{\fairness}{\mathrm{fairness}}

\newcommand{\Mopt}{M^*}
\newcommand{\itemcolor}{black}
\newcommand{\agentcolor}{magenta}

\begin{document}

\title{
Fair metric distortion for matching with preferences
\vspace{15pt}}

\author{Jabari Hastings \\ \hspace{0pt}{Stanford University}\\\hspace{0pt}{\texttt{jabarih@stanford.edu}}
\and Prasanna Ramakrishnan\\ \hspace{0pt}{Stanford University}  \\ \hspace{0pt}{\texttt{pras1712@stanford.edu}}
}

\date{\vspace{15pt}\small{\today}}

\maketitle

\begin{abstract} 
\input{abstract}
\end{abstract}

\thispagestyle{empty}
\newpage 

\section{Introduction}
\input{intro}

\section{Preliminaries}
\input{prelims}

\section{Max-Distortion vs Sum-Distortion}\label{sec:max-vs-sum}
\input{max-vs-sum}

\section{Distortion of RepMatch}
\input{modified-rep-match}\label{sec:mod-rep-match}

\section{Max-Distortion Lower Bounds}\label{sec:lower-bounds}
\input{lower-bounds}

\section{Discussion}
\input{discussion}

\section*{Acknowledgments}
JH is supported by the Simons Foundation Collaboration on the Theory of Algorithmic Fairness and the Simons Foundation investigators award 689988. PR is
supported by Moses Charikar’s Simons Investigator Award and Li-Yang Tan’s NSF awards 1942123,
2211237, 2224246, Sloan Research Fellowship, and Google Research Scholar Award.

\bibliography{papers,ref}
\bibliographystyle{alpha}

\appendix

\section{Missing Details from \texorpdfstring{\Cref{sec:max-vs-sum}}{Section 3}}\label{app:max-vs-sum}
\input{app-max-vs-sum}

\section{Missing Details from \texorpdfstring{\Cref{sec:mod-rep-match}}{Section 4.2}}\label{app:rep-match}
\input{app-rep-match}

\section{Missing Details from \texorpdfstring{\Cref{sec:lower-bounds}}{Section 5}}\label{app:lower-bounds}
\input{app-lower-bounds}

\section{Max-Distortion of Other Mechanisms}\label{app:other-mechanisms}
\input{app-other-mechanisms}

\end{document}

%% file: abstract.tex
We consider the matching problem in the \emph{metric distortion} framework.
There are $n$ agents and $n$ items occupying points in a shared metric space, and the goal is to design a matching mechanism that outputs a low-cost matching between the agents and items, using only agents' ordinal rankings of the candidates by distance. A mechanism has \emph{distortion} $\alpha$ if it always outputs a matching whose cost is within a factor of $\alpha$ of the optimum, in every instance regardless of the metric space. 

Typically, the cost of a matching is measured in terms of the \emph{total} distance between matched agents and items, but this measure can incentivize unfair outcomes where a handful of agents bear the brunt of the cost. With this in mind, we consider how the metric distortion problem changes when the cost is instead measured in terms of the \emph{maximum} cost of any agent. 
We show that while these two notions of distortion can in general differ by a factor of $n$, the distortion of a variant of the state-of-the-art mechanism, \emph{RepMatch}, actually improves from $O(n^2)$ under the sum objective to $O(n^{1.58})$ under the max objective. We also show that for any fairness objective defined by a monotone symmetric norm, this algorithm guarantees distortion $O(n^2)$.

%% file: intro.tex
Matching is a central problem in social choice theory that underpins a wide array of resource allocation scenarios. Just within education, centralized systems match students to schools, dorms, and courses; faculty and TAs to teaching assignments; and courses to classrooms. While the quality of an outcome is often evaluated in terms of the participants' \emph{cardinal} utilities (or costs) in the sense of von Neumann and Morgenstern \cite{von1947theory}, it can often be difficult for agents to quantify their utility and express cardinal preferences. In light of this, a variety of widely adopted mechanisms (such as \emph{Random Serial Dictatorship}, the \emph{Boston mechanism} \cite{AS03}, and \emph{Deferred Acceptance}) operate based on \emph{ordinal} preferences, which only require agents to compare one option against another. 

Within social choice, one particularly fruitful approach towards reconciling the disconnect between cardinal utilities and ordinal preferences is the \emph{distortion} framework, introduced by the seminal work of  Procaccia and Rosenschein \cite{DBLP:conf/cia/ProcacciaR06} (see \cite{DBLP:conf/ijcai/AnshelevichF0V21} for a detailed survey). Under the this framework, a mechanism is evaluated in terms of how well it approximates the optimal outcome, in the worst case over all possible ordinal preferences, and all cardinal utilities consistent with these preferences. Perhaps unsurprisingly, every mechanism can be arbitrarily worse than the optimal solution without any limitations on the utilities. However, by adding relatively mild constraints, it becomes tractable to find mechanisms with strong, meaningful guarantees.

Along these lines, substantial research has focused on the \emph{metric distortion} problem, originally introduced by Anshelevich, Bhardwaj, and Postl \cite{ABP15} in the context of voting. 
In this model, we imagine that the various voters and candidates  occupy points in an unknown metric space, and each voter's cost of a candidate is the distance between them (i.e., voters prefer candidates that are close). The objective is to find mechanisms (\emph{voting rules}) that, using only each voter's ordinal ranking of candidates by distance, always choose candidates whose \emph{total} cost is only worse than the best possible by a small factor (called the \emph{distortion}). 
Though the metric distortion model was originally motivated by well-studied spatial models of voting \cite{enelow1984spatial,enelow1990advances,merrill1999unified,armstrong2020analyzing}, the power of the model comes not from a conviction that preferences are truly represented by some metric, but rather from the worst-case guarantee that a low-distortion mechanism must find a solution that is near-optimal simultaneously in every metric space. Despite the challenging nature of this task, a number of voting rules can achieve distortion guarantees that are small constants, and a long line of work managed to pin down the optimal distortion of deterministic rules at just 3 \cite{ABP15, Goel:2017aa,MW19, Kem20b, GHS20,KK22,DBLP:conf/sigecom/Kizilkaya023}. Through this work, the metric distortion model has proven to be a powerful lens not just to gain a better understanding of existing rules (as in \cite{ABP15}), but also to motivate the design of natural new voting rules (as in \cite{MW19,KK22,DBLP:conf/sigecom/Kizilkaya023,Charikar:2024aa}). Metric distortion has even played a role as a kind of sandbox to understand the power of randomization \cite{DBLP:journals/jair/AnshelevichP17,DBLP:conf/aaai/FainGMP19,Kem20a,DBLP:conf/soda/CharikarR22,Charikar:2024aa}, and entirely different paradigms of voting and preference elicitation \cite{DBLP:conf/aaai/GrossAX17,DBLP:conf/sigecom/ChengDK17,DBLP:conf/aaai/ChengDK18,DBLP:journals/ai/CaragiannisSV22,Goyal:2024aa, Goel:2025aa}. 

Even though matching is arguably just as ubiquitous a social choice problem as voting, it is significantly less understood within the metric distortion problem. In this setting, instead of voters and candidates we have $n$ agents and $n$ items, and rather than choosing a single winning candidate, the goal is to match the agents to items. Once again, agents' costs of items are derived from a metric space. The question is whether matching mechanisms can, with only the agent's ordinal preferences over items, choose a matching whose total cost is close to the minimum possible cost. 

The initial work of Caragiannis et al. \cite{CFRF+16} focused on the metric distortion of \emph{Serial Dictatorship} mechanisms, which iterate over the agents and match them to their favorite unmatched item. They showed that with a deterministically chosen order, this mechanism can have  distortion $2^n - 1$, but with a uniformly random chosen order (\emph{Random Serial Dictatorship}), its distortion is at most $O(n)$ in expectation. This was followed by the work of Anari, Charikar, and Ramakrishnan \cite{Anari:2023aa}, who designed a simple and novel deterministic matching mechanism with distortion at most $O(n^2)$ called \emph{RepMatch}. They also showed that any mechanism, even with randomization, must have worst-case distortion at least $\Omega(\log n)$, dashing hopes of constant distortion mechanisms like in the voting setting. Despite the exponential gap between the best lower and upper bounds, further progress on the problem has been elusive.
A recent work of Filos-Ratsikas et al. \cite{Filos-Ratsikas:2025aa} has instead studied the distortion in simple metric spaces like the real line,  but our understanding of the distortion in general metric spaces is still fairly limited.  

While it remains reasonable to hope that resolving the metric distortion problem for matching will continue to lead to the discovery of natural new matching mechanisms, a fair critique is that this work takes for granted that the ``optimal'' choice is the one that minimizes the total social cost. For example, if the goal is to design a school choice mechanism such that students are matched to nearby schools, it would be difficult to argue that a matching that reduces a dozen students' commute by 5 minutes but increases one student's commute by an hour is just as good. From a fairness perspective, it might be more natural to aim to minimize the \emph{egalitarian} social cost, which is the maximum cost incurred by any agent rather than the \emph{utilitarian} social cost, which is the sum. 

As it turns out, in the voting setting this criticism is definitively addressed by \cite{Goel:2017aa,Goel:2018aa}. 
Goel, Hulett, and Krishnaswamy \cite{Goel:2018aa} show that any candidate that has distortion $\alpha$ with respect to the utilitarian objective (simply, \textit{sum-distortion} $\alpha$) must also have distortion at most $\alpha + 2$ with respect to the egalitarian objective (\textit{max-distortion} $\alpha + 2$). In fact, such a candidate has distortion at most $\alpha + 2$ with respect to any objective which is a monotone symmetric norm of the costs (called the \textit{fairness ratio}). (Here, the $\ell_1$ and $\ell_\infty$ norms would correspond to the utilitarian and egalitarian objectives respectively.) 
As a result, guaranteeing low distortion with respect to the sum objective automatically ensures low distortion across a broad family of fairness objectives, each differing in how much emphasis it places on the worst-off agents versus the group as a whole. 

Our central question is whether the same holds true in the matching setting.

\begin{question}\label{q:main}
Do matching mechanisms with low sum-distortion also have low max-distortion and fairness ratio?  
\end{question}

\subsection{Our Contributions} 
In this work, we make the first steps towards answering \Cref{q:main}. To begin with, we show using simple examples that low sum distortion \emph{does not} imply low max distortion (\Cref{prop:small-sum-large-max}) or vice versa (\Cref{prop:small-max-large-sum})--- the gap between the two can be a factor of $n$. This observation sits in stark contrast to the voting setting, and suggests that there is value in analyzing the metric distortion of matching mechanisms beyond just the sum objective.

Following along these lines, we consider the \emph{RepMatch} mechanism \cite{Anari:2023aa}, which has the state-of-the-art (deterministic) sum-distortion guarantee of $O(n^2)$. We analyze the max-distortion and fairness ratio of  RepMatch with a minor modification to its tie-breaking procedure (detailed in \Cref{alg:rep-match}) that simplifies the algorithm and its analysis.
Naively, the link  between distortion and the fairness ratio (given by \Cref{prop:distortion-fairness-bounds}) suggests that the max-distortion of this mechanism could be as large as $\Omega(n^3)$. However, we show that this is not the case.

\begin{theorem}\label{thm:main}
\Cref{alg:rep-match} guarantees a max-distortion $O(n^{1.58})$ and fairness ratio $O(n^2)$. 
\end{theorem}

That is, we show that with a careful analysis, the distortion of RepMatch actually \emph{improves} for the max objective by a polynomial factor. Not only this, but we can generalize the analysis to any fairness objective that is a monotone symmetric norm. The sum-objective ends up being the worst case, despite the fact that the mechanism was specifically designed with that objective in mind. 

Qualitatively, we view \Cref{thm:main} in an optimistic light. While it is not possible to simply lift a sum-distortion guarantee to a max-distortion or fairness ratio guarantee, it seems plausible that these objectives are still aligned with each other in some sense. We also view max-distortion as another avenue for finding new preferential matching mechanisms, with the hope that they will also satisfy strong sum-distortion and fairness ratio guarantees. 

Finally, we complement these results with an exploration of lower bounds. We show that \Cref{alg:rep-match} can have sum-distortion (and therefore fairness ratio) as large as $\Omega(n^2)$, and max-distortion as large as $\Omega(n)$ (\Cref{thm:mod-rep-match-lb-topk}). We also show that every matching mechanism has max-distortion at least $7$ (\Cref{thm:max-distortion-lower-bounds}), using constructions similar to the binary tree metrics introduced by \cite{Anari:2023aa}. While \cite{Anari:2023aa} are able to get a lower bound of $\Omega(\log n)$ for sum-distortion using these examples, curiously we show that for max-distortion, they go no further than $7$ (\Cref{prop:ub-tree}). This fact opens up an intriguing possibility --- that there could be matching mechanisms that have constant max-distortion, as was possible in the voting setting. On the other hand, finding a mechanism with sum or max distortion just $o(n)$ remains a challenging open problem.

\subsection{Related work}
The problem of matching under incomplete information was introduced by \cite{HZ79} and has since been studied in a variety of contexts. Distortion in ordinal matching was first studied by
\cite{FRFZ14}. Instead of considering costs in a metric space, they consider \emph{unit-sum} and \emph{unit-range} valuations. Note that by focusing on valuations rather than costs, their setting becomes one of maximization. \cite{FRFZ14} showed that the best distortion attainable by a randomized mechanism in their setting is $O(\sqrt{n})$ and this is achieved by the Random Serial Dictatorship.
\cite{ABFRV22b} later showed that for unit-sum valuations, distortion $\Theta(n^2)$ is the best possible for deterministic mechanisms.

Another interesting setting permits unrestricted valuation functions, but gives the mechanism the power to query a small number of agent valuations.
\cite{ABFRV22b} showed that with $O(\lambda \log n)$ queries per agent, it is possible to get distortion $O(n^{1/(\lambda + 1)})$. Even with just two queries per agent, \cite{ABFRV22a} showed that one can get distortion $O(\sqrt{n})$. \cite{Ebadian:2025aa} ultimately showed that with $\lambda$ queries per agent, a distortion of $O(n^{1/\lambda})$ is attainable, which is optimal.

Our work is the first in the metric matching setting to study the fairness ratio within arbitrary metric spaces. We remark that the recent work of \cite{Filos-Ratsikas:2025aa} gives a mechanism with fairness ratio of 3 for line metrics, but the guarantees do not to extend to more general spaces. 

Metric assumptions have also been applied to maximum matchings, as opposed to our minimization setting. Here, the best distortion bounds are constant; \cite{AS16p,AS16q} gave a mechanism with distortion $1.6$ and a truthful one that with distortion $1.76$ (both of which are randomized). We also note that \cite{AZ21} study minimum matchings in a simpler setting, where the locations of the agents are unknown, but the locations of the items are known. 
With this additional information, they show that deterministic mechanisms can attain optimal distortion of $3$.

%% file: prelims.tex
\paragraph{Notation and terminology.} 
Throughout the paper we let $A = \{a_1, a_2, \dots, a_n\}$  denote the set of agents, and $B = \{b_1, b_2, \dots, b_n\}$ denote the set of items, where $n\geq 2$ is some integer. An \emph{ordinal matching problem instance}, denoted by $\sigma$, is given by $n$ \emph{preference lists}, one provided by each agent. Each preference list is a permutation of the items, given in non-decreasing order of distance. We say that agent $a_i$ \emph{prefers} item $b_j$ \emph{over} item $b_k$ if the item $b_j$ appears before $b_k$ in her preference list. Agent $a_i's$ \emph{favorite} item in a set $S \subseteq B$ is the item in $S$ that appears the earliest in her preference list.
Similarly, agent $a_i$'s favorite $t$ items are the first $t$ items in her preference list. 

\paragraph{Metric spaces.} The distances between agents and items are conveyed by an unknown \emph{pseudometric} \allowbreak $d:(A\cup B) \times (A\cup B) \to \mathbb{R}_{\geq 0}$. 
By definition, the function $d$ satisfies the following  properties.

\begin{enumerate}[label=(\roman*)]
	\item Identity of indiscernibles: $d(x, y) = 0$ if $x = y$, 

	\item Symmetry: $d(x, y) = d(y, x)$,

	\item Triangle inequality: $d(x, y) \leq d(x, z) + d(z, y)$.

\end{enumerate}
In certain definitions of metric spaces, it is required that different elements have positive distance (i.e. $d(x, y) = 0$ if and only if $x = y$), but we do not require this. We allow different agents or items to ``occupy'' the same point in the metric space, in which case their distance is $0$; however, if they are located at different points, they do indeed have positive distance between each other. We note that we can easily replace distances of $0$ with sufficiently small $\epsilon$'s without changing any results.

We say that a metric $d$ is \emph{consistent} with an instance $\sigma$ if for every agent $a_i$, whenever that agent  $a_i$ prefers an item $b_j$ over item $b_k$, then  $d(a_i, b_j) \leq d(a_i, b_k)$. We let $\rho(\sigma)$ denote the set of all metrics consistent with the problem instance.

\paragraph{Matchings.} We will use injective functions $M:A \to A \cup B$ to represent matchings. For each agent $a_i$, we can have  $M(a_i) \in B$ when the agent $a_i$ is assigned an item, and $M(a_i) = a_i$ otherwise.
When $M$ is a bijection between $A$ and $B$, we say that it is a \emph{perfect matching}.
For any set of agents $S \subseteq A$ and matching $M$, we will let $M|_S : S \to A \cup B$ denote the \emph{restriction} of the matching $M$ to the set $S$. We assume that $M|_S$ is the identity map on the set $A \setminus S$.

\paragraph{Cost of Matchings.}
Given a \emph{norm} $f: \mathbb R^n \to \mathbb R$, the cost of a matching $M$ under a metric $d$ is 
\begin{equation*}
\cost_f(M, d) \coloneqq f \left ( d(a_1, M(a_1)), d(a_2, M(a_2)), \dots, d(a_n, M(a_n)) \right ).
\end{equation*}
When the context is clear, we will simplify our notation and write $\cost_f(M)$ instead of $\cost_f(M, d)$
With a slight abuse of notation, we  also define the top-$k$ cost of a matching $M$ to be the sum of the $k$ largest distances between agents and their assigned items,
\begin{equation*}
\ctopk(M, d) \coloneqq \max_{S \subseteq A, |S| = k} \sum_{a_i \in S} d(a_i, M(a_i)). 
\end{equation*}
Notice that the functions $\cost_1(M, d), \cost_n(M, d)$ compute the maximum distance and sum of distances respectively.  

\paragraph{Distortion.}
Given a norm $f$ and metric $d$, the \emph{$f$-distortion} of a perfect matching $M$ is the ratio of that matching's  cost compared to that of the optimal matching,
\begin{equation*}
\distortion_f(M, d) := \frac{\cost_f(M, d)}{\min_{\Mopt} \cost_f(\Mopt, d)}.
\end{equation*}
We will refer to the top-$1$ distortion top-$n$ distortion  as the max-distortion and sum-distortion respectively.
The worst-case $f$-distortion of a matching is given by  
\begin{equation*}
\distortion_f(M, \sigma) := \sup_{d \in \rho(\sigma)} \distortion_f(M, d).
\end{equation*}
With a slight abuse of notation, the top-$k$ distortion is defined in a similar fashion.

\paragraph{Fairness Ratio.}
Borrowing the terminology introduced by \cite{Goel:2018aa}, we also define the \emph{fairness ratio} of a perfect matching $M$ to be maximum worst-case top-$k$-distortion
\begin{equation*}
    \fairness(M, d) := \max_{1 \leq k \leq n} \distortion_k(M, d).
\end{equation*}
The worst-case fairness ratio of a matching is given by 
\begin{equation*}
\fairness(M, \sigma) := \sup_{d \in \rho(\sigma)} \fairness(M, d).
\end{equation*}
One of the main reasons for studying the fairness ratio is that it gives an approximation result for the class of monotone symmetric norms.

\begin{proposition}[\cite{Goel:2017aa}, Theorem 2.3]
If a matching $M$ is such that $\fairness(M, \sigma) \leq \alpha$ for some problem instance $\sigma$ and some $\alpha \in \mathbb R$, then for every monotone symmetric norm $f$, $\distortion_f(M, \sigma) \leq \alpha$.
\end{proposition}

%% file: max-vs-sum.tex
In the voting setting, \cite{Goel:2018aa} showed that in any metric that is consistent with a given problem instance, a mechanism's max-distortion cannot exceed its sum-distortion by more than an additive constant\footnote{Note that \cite{Goel:2018aa} also go on to prove the stronger result that for any \emph{instance}, the difference between worst-case fairness ratio and worst-case sum-distortion is bounded by a constant.}. This is no longer the case in the matching setting. In particular, we can show that there are metrics in which a matching's max-distortion and sum-distortion can differ by as much as a multiplicative factor of $n$.
\begin{proposition}\label{prop:small-sum-large-max}
There exists a metric $d$ and a perfect matching $M$ such that
\begin{equation*}
(n/2) \cdot \distortion_n(M, d) \leq \distortion_1(M, d).
\end{equation*}
\end{proposition}

\begin{proposition}\label{prop:small-max-large-sum}
There exists a metric $d$ and a perfect matching $M$ such that
\begin{equation*}
n \cdot \distortion_1(M, d) \leq \distortion_n(M, d).
\end{equation*}
\end{proposition}

Taken together,  \Cref{prop:small-sum-large-max,prop:small-max-large-sum}
suggest that a bound on the sum-distortion may not automatically translate to a similar bound (up to constants) on the max-distortion and vice versa; a more careful analysis of the \emph{worst-case} sum-distortion and max-distortion is likely needed.
Nevertheless, both quantities are somewhat useful for bounding the fairness ratio, as we show below.
\begin{proposition}\label{prop:distortion-fairness-bounds}
For any metric $d$ and any perfect matching $M$,
\begin{align*}
 \fairness(M, d) \leq n \cdot \min ( \distortion_1(M, d) , \distortion_n(M, d) ).
\end{align*}
Furthermore, for any problem instance $\sigma$ and any perfect matching $M$,
\begin{align*}
 \fairness(M, \sigma) \leq n \cdot \min ( \distortion_1(M, \sigma) , \distortion_n(M, \sigma) ).
\end{align*}
\end{proposition}
In the absence of a bound on the worst-case top-$k$-distortion, it can therefore be worthwhile to consider both the worst-case max-distortion and sum-distortion when analyzing the worst-case fairness ratio. For proofs of the above results, see \Cref{app:max-vs-sum}.

%% file: modified-rep-match.tex
The RepMatch mechanism, proposed by \cite{Anari:2023aa}, is a simple procedure that relies on the following intuitive premise. If two agents have similar favorite items, then they are likely to be close to each other in some sense, and thus a good assignment for one agent is also likely to be a good assignment for the other. The mechanism iteratively clusters
agents whose distance to each other is small, and uses the preferences of a single designated agent within the cluster, called the \textit{representative},  as a proxy for the others' preferences. If the representatives of two clusters overlap in their favorite items, then we can merge the clusters, and promote one of the representatives to become the representative of the merged cluster.
Once the representatives' favorite items do not overlap, the mechanism arbitrarily matches each agent to one of their representative's favorite items.

One design choice in the mechanism is in deciding which representative gets promoted when merging two clusters. \cite{Anari:2023aa} do this by associating each cluster $S$ with a parameter $\lambda$, and promoting the representative corresponding to the larger $\lambda$. Their analysis shows that the total cost of matching the agents in $S$ to the top $|S|$ choices of the representative is at most $|S|\cdot 2^\lambda$ times their cost in the optimal matching, and at the same time $2^\lambda \leq |S|$, giving a bound of $O(n^2)$ for the sum-distortion.

We take a more direct approach, by simply promoting the representative whose cluster is larger. We show that this modification still retains the sum-distortion guarantee, but also allows the analysis to give improved results for the max-distortion and more generally top-$k$-distortion.

Below, we present the formal description of the mechanism. 

\begin{center}
\myalg{alg:rep-match}{RepMatch (with size-based promotion)}{
Initially, there are $n$ singleton sets $S_1, S_2, \dots, S_n$, with $S_i = \{a_i\}$, $r_i = a_i$.\\
While there exist two sets $S_i$ and $S_j$ such that some item is both one of $r_i$'s favorite $|S_i|$ items and one of $r_j$'s favorite $|S_j|$ items: 
\begin{itemize}
	\item Replace $S_i$ and $S_j$ with $S_i \cup S_j$. 
	\item The representative for $S_i \cup S_j$ is $r_i$ if $|S_i| \geq |S_j|$ and $r_j$ otherwise. 
\end{itemize} 
Arbitrarily match the agents in each $S_i$ to the favorite $|S_i|$ choices of $r_i$.
}
\end{center}

Athough the RepMatch mechanism was designed with the metric distortion framework in mind, it is quite natural and its description makes no mention of metric spaces. It seems reasonable for a representative to act on the behalf of other individuals with similar preferences. When two representatives share similar interests, it might also make sense for them to form a larger coalition.

Since the worst-case sum-distortion of the RepMatch mechanism was shown by \cite{Anari:2023aa} to be $O(n^2)$, a naive application of  \Cref{prop:distortion-fairness-bounds} would have suggested a fairness ratio of $O(n^3)$. However, a careful analysis of \Cref{alg:rep-match} presented in the following sections provides an improved bound for the fairness ratio and new $f$-distortion bounds for several norms $f$.

\subsection{Upper Bounds}
We begin by providing an upper bound on the top-$k$-distortion of \Cref{alg:rep-match}.
\begin{theorem}\label{thm:mod-rep-match-ub-topk}
\Cref{alg:rep-match} guarantees  worst-case top-$k$-distortion  $O(k^2 (n/k)^{\log_2 3})$. 
\end{theorem}
Note that our upper bound on the top-$k$-distortion immediately translates to an upper bound on the fairness ratio. Since  $k^2 (n/k)^{\log_2 3} \leq n^2$  for every positive integer $k \leq n$, we obtain a quadratic upper bound on the $f$-distortion for every monotone symmetric norm $f$.
Setting $k = 1$, our result also implies a sub-quadratic upper bound on the max-distortion.
\begin{corollary}
\Cref{alg:rep-match} guarantees  worst-case fairness ratio $O(n^2)$.
\end{corollary}

\begin{corollary}
\Cref{alg:rep-match} guarantees  worst-case max-distortion  $O(n^{\log_2 3})$.   
\end{corollary}

\begin{proof}[Proof of \Cref{thm:mod-rep-match-ub-topk}]
Fix a consistent metric $d$ and let $\Mopt_k$ be the perfect matching that minimizes the function $\ctopk$. For the sake of the analysis, we associate each set $S$ of agents maintained by the algorithm  with a nonnegative ``weight'' $w(S)$. For the initial singletons $S_i$ we have $w(S_i) = 1$. Whenever the mechanism merges sets $S_i$ and $S_j$ with $|S_i| \geq |S_j|$, the weight associated with the union $S_i \cup S_j$ is as follows,
\begin{equation}\label{eq:mod-rep-match-topk-inv1}
w(S_i \cup S_j) = 
\begin{cases}
\max(2w(S_j), w(S_i)) & |S_i| + |S_j| \leq k \\ 
2w(S_j) + w(S_i) & |S_i| + |S_j| > k
\end{cases}.
\end{equation}
We will show that the mechanism maintains the invariant that for each set $S_i$ and each agent $a \in S_i$, 
\begin{equation}\label{eq:mod-rep-match-topk-inv2}
d(a, r_i) \leq -d(a, \Mopt_k(a)) + w(S_i) \cdot \ctopk(\Mopt_k|_{S_i}).
\end{equation}
Note that this implies that for each agent $a \in S_i$, 
\begin{equation*}
d(r_i, \Mopt_k(a))
\leq d(r_i, a) + d(a, \Mopt_k(a))   \leq w(S_i) \cdot \ctopk(\Mopt_k|_{S_i}).
\end{equation*}
Moreover, it follows that there are at least $|S_i|$ distinct items $b$ such that
\begin{equation}\label{eq:mod-rep-match-topk-inv3}
d(r_j, b) \leq w(S_i) \cdot \ctopk(\Mopt_k|_{S_i}).
\end{equation}
In particular, this holds when $b$ is one of $r_i$'s favorite $|S_i|$ items. If an agent $a \in S_i$ is matched by the mechanism to one of those favorite $|S_i|$ items, then by \eqref{eq:mod-rep-match-topk-inv2} and \eqref{eq:mod-rep-match-topk-inv3} she will incur a cost at most 
\begin{align*}
d(a, r_i)  + d(r_i, b) 
\leq 2 w(S_i) \cdot \ctopk(\Mopt_k|_{S_i}) 
\leq 2 w(S_i) \cdot \ctopk(\Mopt_k).
\end{align*}
We will additionally show that for every set $S_i$ maintained by the mechanism, $w(S_i) \leq k(n/k)^{\log_2 3}$. This, coupled with the above, implies that sum of the $k$ largest costs incurred as a result of the mechanism is at most $2k^2  (n/k)^{\log_2 3} \cdot \ctopk(\Mopt_k)$, which in turn implies the stated distortion bound. 

\paragraph{Establishing the invariant.}
Clearly the invariant holds initially. 
So, we consider what happens at the end of an iteration given that the invariant held at the start of that iteration. During an iteration, there will be some some set $S_i$ that is merged with the set $S_j$.
All other sets remain unchanged, so we only need to show that the invariant holds for the union $S_i \cup S_j$.
We assume without loss of generality that $|S_i| \geq |S_j|$, and thus the agent $r_i$ becomes the representative for the union $S_j \cup S_j$.
If $|S_i \cup S_j| > k$, note that equation \eqref{eq:mod-rep-match-topk-inv1} implies that 
\begin{align*}
w(S_i \cup S_j) \cdot  \ctopk(\Mopt_k|_{S_i \cup S_j})
&= (2w(S_j) + w(S_i) ) \cdot \ctopk(\Mopt_k|_{S_i \cup S_j})\\
&\geq  2w(S_j) \cdot \ctopk(\Mopt_k|_{S_j}) + w(S_i) \cdot \ctopk(\Mopt_k|_{S_i}).
\end{align*}
If $|S_i \cup S_j| \leq k$, note that equation \eqref{eq:mod-rep-match-topk-inv1} also implies that 
\begin{align*}
w(S_i \cup S_j) \cdot  \ctopk(\Mopt_k|_{S_i \cup S_j})
&= \max(2w(S_j), w(S_i)) \cdot \ctopk(\Mopt_k|_{S_i \cup S_j})\\
&= \max(2w(S_j), w(S_i)) \cdot (\ctopk(\Mopt_k|_{S_i}) + \ctopk(\Mopt_k|_{S_j}))\\
&\geq  2w(S_j) \cdot \ctopk(\Mopt_k|_{S_j}) + w(S_i) \cdot \ctopk(\Mopt_k|_{S_i}).
\end{align*}
Notice that we have an identical lower bound in both cases.

Since $S_i$ and $S_j$ are merged, there must be some item $b$ that is one of representative $r_i$'s favorite $|S_i|$ items and one of representative $r_j$'s favorite $|S_j|$ items. This observation, along with equation \eqref{eq:mod-rep-match-topk-inv2}, implies that  for every agent $a \in S_j$,
\begin{align*}
d(a, r_i) 
&\leq d(a, r_j) + d(r_j, b) + d(b, r_i) \\
&\leq -d(a, \Mopt_k(a)) + 2w(S_j) \cdot \ctopk(\Mopt_k|_{S_j}) + w(S_i) \cdot \ctopk(\Mopt_k|_{S_i}). \\
&\leq -d(a, \Mopt_k(a)) + w(S_i \cup S_j) \cdot  \ctopk(\Mopt_k|_{S_i \cup S_j})
\end{align*}
as desired. 
By equations \eqref{eq:mod-rep-match-topk-inv1} and \eqref{eq:mod-rep-match-topk-inv2}, we can also  show that for every agent $a \in S_i$,
\begin{align*}
d(a, r_i) 
&\leq -d(a, \Mopt_k(a)) + w(S_i) \cdot \ctopk(\Mopt_k|_{S_i})  \\
&\leq -d(a, \Mopt_k(a)) + w(S_i \cup S_j) \cdot \ctopk(\Mopt_k|_{S_i \cup S_j}), 
\end{align*}
as desired. 
Therefore, the invariant is maintained throughout the mechanism.

\paragraph{Bounding the growth of $w(S)$.}
It remains to determine how the weights grow as the mechanism merges sets. We claim that for each set $S$ maintained by the the mechanism, 
\begin{equation}\label{eq:mod-rep-match-topk-weights-growth}
w(S) 
\leq 
\begin{cases}
|S| & \text{if } |S| \leq k\\
k(|S|/k)^{\log_23} & \text{if } |S| > k
\end{cases}.
\end{equation}
We will show this inductively. As our base case, we consider when $|S| \leq k$. It is not hard to verify that in this case, we have $w(S) \leq |S|$.
This clearly holds when the set $S$ is a singleton set. If the set $S$ is the result of  merging the sets $S_i$ and $S_j$ with $k \geq |S_i| \geq |S_j|$ and weights $w(S_i)$ and $w(S_j)$ satisfying equation \eqref{eq:mod-rep-match-topk-weights-growth}, then equation \eqref{eq:mod-rep-match-topk-inv1} implies that 
\begin{equation*}
w(S_i \cup S_j)
= \max ( 2w(S_j) , w(S_i))
\leq \max (2 |S_j| , |S_i|) 
\leq |S_j| + |S_i|,
\end{equation*}
as desired. 

Now, suppose that $|S| > k$ and that $S$ is the result of merging sets $S_i$ and $S_j$ with $|S_i| \geq |S_j|$ and weights $w(S_i)$ and $w(S_j)$ satisfying equation \eqref{eq:mod-rep-match-topk-weights-growth}. Note that if a set $T \in \{S_i, S_j \}$ satisfies $|T| \leq k$, then it also the case that $|T| \leq k (|T|/k)^{\log_2 3}$. This, along with equation \eqref{eq:mod-rep-match-topk-inv1} implies that 
\begin{align*} 
w(S_i \cup S_j) 
&= 2w(S_j) + w(S_i) \\
&\leq 2k(|S_j|/k)^{\log_23} + k(|S_i|/k)^{\log_2 3} \\
&\leq k(|S_i\cup S_j|/k)^{\log_2 3},
\end{align*}
where the final inequality uses the fact that for $0 < x \leq y$ and $c = \log_2 3$, $2x^c + y^c \leq (x + y)^c$, which is proven in \Cref{app:rep-match} for completeness. Since equation \eqref{eq:mod-rep-match-topk-weights-growth} continues to be satisfied, our induction is complete.
\end{proof}

\subsection{Lower Bounds}
We also provide lower bounds for the top-$k$-distortion of \Cref{alg:rep-match}. The lower bounds that we establish will assume \emph{worst-case} tie-breaking at the points where arbitrary decisions can be made. Note that there can be several of these points in the algorithm. For instance, when two sets $S_i$ and $S_j$ with $|S_i| = |S_j|$ are merged, either agent $r_i$ or agent $r_j$ can be the representative of the union $S_i \cup S_j$, but the mechanism has to choose one of them. Another point of arbitrariness is when the mechanism has to assign agents in a partition $S_i$ to the favorite $|S_i|$ choices of the representative $r_i$. 

Our lower bounds will not rule out the existence of a matching with low distortion among the family of matchings that can be returned by the algorithm. 
In a sense, they only highlight a specific failure mode of the mechanism when tie-breaking is done in an ad-hoc fashion. We suspect that with more careful tie-breaking, one might be able to circumvent these failure modes and improve the distortion of the algorithm.

\begin{theorem}\label{thm:mod-rep-match-lb-topk}
\Cref{alg:rep-match} (with worst-case tie-breaking) can produce a matching with worst-case top-$k$ distortion $\Omega(kn)$.
\end{theorem}
Note that our lower bound on the top-$k$-distortion immediately translates to a linear lower bound on the  max-distortion and a tight, quadratic lower bound on the sum-distortion.

\begin{corollary}\label{cor:mod-rep-match-lb-max-sum}
\Cref{alg:rep-match} (with worst-case tie-breaking) can produce a matching with worst-case max-distortion $\Omega(n)$ and worst-case sum-distortion $\Omega(n^2)$.
\end{corollary}
\begin{proof}[Proof of \Cref{thm:mod-rep-match-lb-topk}]
We will construct a set of agent preferences and a metric consistent with those preferences such that for our specified metric, the optimal perfect matching for the top-$k$ objective, $\Mopt_k$, satisfies $\cost_k(\Mopt_k) = 1$.
We will also show that \Cref{alg:rep-match} can return a matching $M$ such that $\cost_k(M) = \Omega(kn)$.
At a high level, the large distortion arises because inappropriate representatives could be chosen at various points of the mechanism, leading to a domino effect of unnecessary merges.
\paragraph{Construction.}
We construct a family of instances of $n = 2^{\ell} $ agents and items configured along the real line, where $\ell \in \mathbb Z^+$ is some natural number. 
Items are located at the points 
$$\{-1 \}  \cup  \{ 2^{t} - 1  : t \in \mathbb Z^+,  t \leq \ell   \}.$$
The items are partitioned into sets $B_0, B_1, \dots, B_{k}$. 
The set $B_{0}$ consists of a single item located at the point $-1$.  
For every positive integer $t \leq \ell$, the set $B_t$ consists of $2^{t-1}$ items located at the point $2^{t} - 1$. For each nonnegative integer $t$, we will frequently refer to a special item $b_t$ that belongs to the set $B_t$. 
Agents are located at the points 
$$\{0 \}  \cup  \{ 2^{t} - 1  : t \in \mathbb Z^+, t \leq \ell  \}.$$
The agents are partitioned into sets $A_0, A_1, \dots, A_{\ell}$. The set $A_{0}$ consists of a single agent located at the point $0$. 
For every positive integer $t \leq k$, the set $A_t$ consists of $2^{t-1}$ agents located at the point $2^{t} - 1$. For each nonnegative integer $t$, we will frequently refer to a special agent $a_t$ that belongs to the set $A_t$. Note that for every nonnegative integer $t \leq \ell - 1$, each agent $a \in A_{t}$ is equidistant from item $b_{0}$ and any item in $B_{t + 1}$. An example of the instance with $\ell = 3$ is shown in \Cref{fig:mod-rep-match-lb}. 

\begin{figure}[ht]
\centering
\input{rep-match-lb}
\caption{The metric space that defines the preferences in our ordinal matching instance for $\ell= 3$. The red nodes are the agents and the black nodes are the items.}\label{fig:mod-rep-match-lb}
\end{figure}
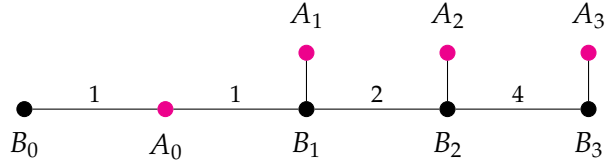

We next describe a set of agent preferences consistent with the metric outlined above. 
For each nonnegative integer $t$, the agents in $A_t$ will have identical preferences, so we focus on outlining the preferences for agents $ a_0, a_1, \dots, a_{\ell}$. 
Every agent will rank items based on distances given by the line metric, resolving ties in favor of items located to her right (i.e., those at a more positive location) and in favor of the special items $b_0, b_1, \dots b_{\ell}$. In particular, for every nonnegative integer $t \leq \ell - 1$ an agent $a_t \in A_t$ will prefer every item in $B_{t+1}$ over item $b_{0}$. Moreover, for every agent $a \in A$, their favorite item in the set $B_t$ will be the special item $b_t$. Note that by construction, the following claim is immediate.

\begin{claim}\label{clm:mod-rep-match-lb-common-favorite}
    For every nonnegative integer $t \leq \ell -1$, one of agent $a_t$'s favorite $\sum_{r = 0}^{t} |A_r|$ items is item $b_{t+1}$. 
\end{claim}

\paragraph{Distortion of \Cref{alg:rep-match}.} It is straightforward to verify that for the instance outlined above, the optimal assignment $\Mopt_1$ for the max-objective assigns agents in the set $A_{t}$ to items in the set $B_{t}$ for every positive integer $t$. Note that $\cost_k(\Mopt_k) = 1$. We will show that it is possible for the algorithm to return a matching $M$ such that for every agent $a \in A_{\ell}$, $M(a) \notin B_{\ell}$. This would imply that  $d(a, M(a)) \geq 2^{\ell-1} $ for every agent $a \in A_{\ell}$. Computing the top-$k$ cost, we have 
\begin{equation}\label{eq:mod-rep-match-lb-topk}
\cost_k(M) \geq  \min(k , |A_\ell|) \cdot 2^{\ell - 1}   \geq \min(k , |A_\ell|) \cdot (n/2)  . 
\end{equation}
Equation \eqref{eq:mod-rep-match-lb-topk} implies the stated distortion lower bound.
In the following claim, we show that the mechanism is indeed able to yield such a perfect matching.
\begin{claim}
For the constructed instance above, a worst-case operation of \Cref{alg:rep-match} could return an arbitrary perfect matching. 
\end{claim}
\begin{proof}
We consider a sequence of merges within \Cref{alg:rep-match} that consists of 2 phases. The first phase is as follows. For each positive integer $t \leq \ell$, repeatedly merge the sets that contain agents from $A_t$ until every agent in $A_t$ is represented by the agent $a_t$. Since the size of every set $A_t$ has cardinality equal to a power of $2$, such a sequence of merges is possible. 

After the first phase of merges is complete, the sets that remain are $A_0, A_1, \dots, A_{\ell}$. The set $A_0$ (with  size $|A_0| = 1$) is represented by agent $a_0$,  and for every positive integer $t \leq \ell$, the set $A_t$ (with size $|A_t| = 2^{t-1}$) is represented by agent $a_t$.

In the second phase, the domino effect kicks in. We first merge the sets represented by agents $a_0$ and $a_1$ and make the representative of the union be agent $a_1$. Then, we merge the sets represented by agents $a_1$ and $a_2$ and make the representative of the union be agent $a_2$, and so on. Such a sequence of merges is possible because: (1) one of agent $a_t$'s favorite $\sum_{r = 0}^{t} |A_r|$ items is the favorite item of agent $a_{t + 1}$ (by  \Cref{clm:mod-rep-match-lb-common-favorite}) and (2) at the time of merging sets represented by agents $a_t$ and $a_{t+1}$, the sets both have size equal to $2^{t}$, so we can make the new representative be agent $a_{t+1}$ and the resulting set will have size equal to $2^{t + 1}$.

Note that at the end of the second phase of merges, every agent becomes represented by the agent $a_{\ell}$. Thus, \Cref{alg:rep-match} could return an arbitrary perfect matching. 
\end{proof}
\end{proof}

%% file: rep-match-lb.tex
\begin{tikzpicture}[scale=0.75]
\tikzstyle{every node}=[draw=black,shape=circle, inner sep=2pt];

\pgfmathsetmacro{\k}{2.5}

\node (b0) at (\k,0) [fill=\itemcolor, color=\itemcolor, label=below:$B_0$] {};
\node (b1) at (3*\k ,0) [fill=\itemcolor, color=\itemcolor, label=below:$B_1$] {};
\node (b2) at (4*\k ,0) [fill=\itemcolor, color=\itemcolor, label=below:$B_2$] {};
\node (b3) at (5*\k ,0) [fill=\itemcolor, color=\itemcolor, label=below:$B_3$] {};

\node (a0) at (2*\k ,0) [fill=\agentcolor, color=\agentcolor, label=below:$A_0$] {};
\node (a1) at (3*\k ,1) [fill=\agentcolor,color=\agentcolor,label=above:$A_1$] {};
\node (a2) at (4*\k ,1) [fill=\agentcolor,color=\agentcolor,label=above:$A_2$] {};
\node (a3) at (5*\k ,1) [fill=\agentcolor, color=\agentcolor, label=above:$A_3$] {};

\path[every node/.style={font=\footnotesize}]
(a0) edge node[yshift=0.2cm] {1} (b0) 
(a0) edge node[yshift=0.2cm] {1} (b1) 
(b1) edge node[yshift=0.2cm] {2} (b2)
(a1) edge node[yshift=0.2cm] {} (b1) 
(a3) edge node[yshift=0.2cm] {} (b3) 
(b2) edge node[yshift=0.2cm] {4} (b3) 
(a2) edge node[yshift=0.2cm] {} (b2);
\end{tikzpicture}

%% file: lower-bounds.tex
Finally, we discuss max-distortion lower bounds for arbitrary deterministic mechanisms.
We are only able to exhibit problem instances that guarantee constant distortion.
\begin{theorem}\label{thm:max-distortion-lower-bounds}
    Every deterministic matching mechanism has worst-case max-distortion at least $7$.
\end{theorem}
Our proof of the above theorem can be found in \Cref{app:lower-bounds}. It is inspired by the approach of \cite{Anari:2023aa}, which was used to give  lower bound of $\Omega(\log n)$ for the sum-distortion. At a high level, their approach consists of two key ingredients: a set of preferences for $n = 2^{k}$ agents given by a tree metric and a family $\mathcal F$ of tree metrics consistent  with those preferences. Every metric  $d \in \mathcal F$ is such that $d(a, b) \in \{0, 1, 2\}$ for any agent $a$ and item $b$. 
Since we desire large max-distortion, we consider families of tree metrics with larger distances.

We now discuss in more detail the structure of the tree metrics that induce agent preferences. For $n \geq 4$, construct a balanced binary tree with the $n$ agents and $n-1$ items, such that the leaf nodes are the agents and the other nodes are items. Let $v$ be the root of this binary tree and let $\ell$ and $r$ be the left and right child respectively. 
Make the remaining item, denoted by $u$, be the parent of the root node $v$.
Note that the nodes $u, v, \ell , r$ are all items. Finally, add  non-negative weights to edges.  We call any set of preferences that arises from this procedure a \emph{tree instance}.

One can easily verify that tree instances in \cite{Anari:2023aa} and the proof of \Cref{thm:max-distortion-lower-bounds} are such that for every agent $a$ in the subtree with root $x \in \{\ell, r \}$, 
\begin{enumerate}[label=(\roman*)]
    \item Agent $a$ prefers each item in the same subtree, as well as item $v$, over item $u$.
    \item Agent $a$ prefers each item in the same subtree over any item in the subtree with root in the set $\{\ell, r \} \setminus \{x\}$. 
\end{enumerate}
In the following proposition, we show that any tree instance with the above properties does not rule out constant max-distortion.
\begin{proposition}\label{prop:ub-tree}
Suppose that preferences are induced by tree instances that satisfy properties (i) and (ii) outlined above.
For any consistent metric $d$, every perfect matching has max-distortion at most $7$.
\end{proposition}
\begin{proof}
Let $\Mopt_1$ be the optimal perfect matching, given the metric $d$. Assume without loss of generality that $\Mopt_1(a_\ell) = u$ for some agent $a_\ell$ in the subtree with root $\ell$. Let $B_{\ell}$ denote the set of items in subtree with root $\ell$ (including $\ell$) and define $B_{r}$ analogously. 
Since $|B_{\ell} \cup \{ v \}| = n/ 2$ and there are $n/2 - 1$ other agents in the same subtree as agent $a_\ell$, there must be some agent $a_r$ in the subtree with root $r$ such that that $\Mopt_1(a_r) \in B_\ell \cup \{ v \}$. 
Property (i) tells us that for every item $b \in B_{\ell} \cup \{ v \}$, 
\begin{equation*}
    d(b, \Mopt_1(a_r) ) \leq d(b, a_\ell) + d(a_\ell, \Mopt_1(a_r)) \leq 2 d(a_\ell, u) \leq  2 \cdot \cost_1(\Mopt_1).
\end{equation*}
By property (ii), we also know that for every item  $b' \in B_{r}$ and every item $b \in B_{\ell}$, 
\begin{align*}
d(b', \Mopt_1(a_r)) 
&\leq d(b', a_r) + d(a_r,\Mopt_1(a_r)) \\
&\leq d(b, a_r) +  d(a_r,\Mopt_1(a_r))  \\
&\leq ( d(b, \Mopt_1(a_r)) + d(\Mopt_1(a_r), a_r ) ) +  d(a_r,\Mopt_1(a_r)) \\
&\leq 4 \cdot \cost_1(\Mopt_1).
\end{align*}
By the consequences of both properties, we find that the distance between any two items is at most $6 \cdot \cost_1(\Mopt_1)$.
Now, let $M$ be an arbitrary perfect matching. Observe that for any agent $a$,
\begin{equation*}
d(a, M(a)) \leq d(a, \Mopt_1(a) ) +  d(\Mopt_1(a), M(a)) \leq 7 \cdot \cost_1(\Mopt_1).
\end{equation*}
Hence, we have $\cost_1(M) \leq 7 \cdot \cost_1(\Mopt_1)$, which implies the stated distortion bound.
\end{proof}

%% file: discussion.tex
Since a mechanism's sum-distortion and max-distortion could differ by as much as a factor of $n$, we believe that it is worthwhile to study both. Note that significant improvements for one could have immediate consequences for the other, as well as the fairness ratio. Thus, a milestone on the path to answering open problems raised by \cite{DBLP:conf/ijcai/AnshelevichF0V21} and \cite{Anari:2023aa} is the following question.
\begin{openproblem}
What is the optimal max-distortion for deterministic mechanisms for metric matching with ordinal preferences?
\end{openproblem}
Our analysis of \Cref{alg:rep-match} mechanism shows that the optimal max-distortion is at most $O(n^{1.58})$.
However, it seems plausible to us that sublinear max-distortion is attainable and we  boldly conjecture that the optimal distortion is in fact a constant.

\begin{conjecture}
    There is a deterministic matching mechanism that guarantees max-distortion $O(1)$.
\end{conjecture}
As a first step towards resolving this conjecture, we suggest the following weaker problem. Suppose that we know that each item is distance at most $1$ to \textit{some} unique agent. For each item, can we find an agent whose distance is at most $O(1)$? If the answer is no, this would immediately imply a super-constant lower bound for max-distortion. If the answer is yes, then perhaps one can use this knowledge to construct better mechanisms.

Another interesting direction relates to  truthfulness. The Serial Dictatorship has exponentially large max-distortion (see \Cref{app:other-mechanisms} for more details), but are there other truthful deterministic mechanisms with better distortion guarantees?

%% file: app-max-vs-sum.tex
\begin{proof}[Proof of \Cref{prop:small-sum-large-max}] We construct a family of instances of $n \geq 2$ agents and items configured along the real line, parametrized by a some constant $k > 1$. 
Agents are located at the points 
$$ \{ (k + 1) t : t \in \mathbb Z^+, t \leq n \}.$$
For each positive integer $t \leq n$, let $a_t$ be located at the point $(k+1)t$.
Items are located at the points 
$$ \{ 1 + (k + 1) t : t \in \mathbb Z^+, t \leq n  \}.$$
For each positive integer $t \leq n$, let $b_t$ be located at the point $1 + (k+1)t$.

An example of the instance with $n = 3$ is shown in \Cref{fig:small-sum-large-max}.  
\begin{figure}[ht]
\centering
\input{small-sum-large-max.tex}
\caption{The metric space that defines the preferences in our ordinal matching instance for $n = 3$. The red nodes are the agents and the black nodes are the items.}\label{fig:small-sum-large-max}
\end{figure}
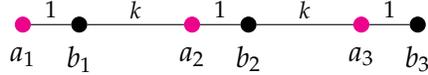

It is straightforward to verify that
the optimal assignment for the max-objective, $\Mopt_1$, satisfies $\cost_1(\Mopt_1) = 1$,
and the optimal assignment for the sum-objective, $\Mopt_n$, satisfies $\cost_n(\Mopt_n) = n$;
take the matching that assigns each agent $a_t$ to item $b_t$.
Consider the matching $M$ that assigns agent $a_t$ to item $b_{t-1}$ for each integer $t$ such that $2 \leq t \leq n$, and assigns agent $a_1$ to item $b_n$. Note that $d(a_t, M(a_t)) = k$ for every integer $t$ such that $2 \leq t \leq n$,  and $d(a_1, M(a_1)) = (n-1)(k+1) + 1$. Thus, we have 
\begin{align*}
\cost_n(M) 
&= (n - 1)k + ((n-1)(k+1) + 1) \\
&= n(2k + 1) - 2k \\
&= nD,
\end{align*}
where $D = (2k + 1) -2k/n$. Note that a simple computation implies that $k = n(D-1)/ (2(n-1))$.
We can also show that 
\begin{align*}
\cost_1(M) 
=  ((n-1)k + n) 
= k(n -1) + n 
= \frac{n (D + 1)}{2}. 
\end{align*} 
\end{proof}

\begin{proof}[Proof of \Cref{prop:small-max-large-sum}] We construct a family of instances of $n$ agents and items configured along a polygon with $n + 1$ vertices, parametrized by some constant $D > 1$. 
The vertices of the polygon (in clockwise order) consist of agents $a_1, a_2 \dots, a_{n}$ followed by item $b_1$. For each positive integer $t \leq n - 1$, the edge $(a_t, a_{t+1})$ has length $D$.
The edge $(a_n, b_1)$  has length $D$ and the edge $(a_1, b_1)$ has length $1$. For every positive integer $t$ such that $2 \leq t \leq n$, the item $b_t$ and agent $a_t$ are collocated. We define the distance between any two vertices to be the length of the shortest path (along the polygon) between them. An example of the instance with $n = 3$ is shown in \Cref{fig:small-max-large-sum}.

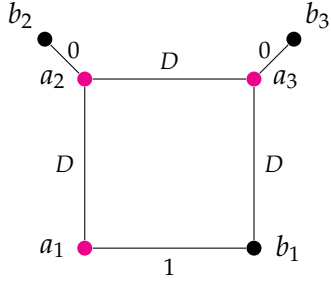
\begin{figure}[ht]
\centering
\input{small-max-large-sum.tex}
\caption{The metric space that defines the preferences in our ordinal matching instance for $n = 3$. The red nodes are the agents and the black nodes are the items.}\label{fig:small-max-large-sum}
\end{figure}

It is straightforward to verify that
the optimal assignment for the max-objective, $\Mopt_1$, satisfies $\cost_1(\Mopt_1) = 1$,
and the optimal assignment for the sum-objective, $\Mopt_n$, satisfies $\cost_n(\Mopt_n) = 1$; take the matching that assigns agent $a_t$ to item $b_t$.
Consider the matching $M$ that assigns agent $a_t$ to item $b_{t+1}$ for each positive integer $t \leq n - 1$, and assigns agent $a_n$ to item $b_1$. It is straightforward to verify that $d(a_t, M(a_t)) = D$ for every positive integer $t \leq n$. Therefore, we have $\cost_1(M) = D$ and $\cost_n(M) = nD$. \qedhere
\end{proof}

\begin{proof}[Proof of \Cref{prop:distortion-fairness-bounds}]

Fix a matching $M$ and metric $d$. Let $\Mopt_1$, $\Mopt_k$ and $\Mopt_n$ be the optimal perfect matchings for the max-objective, top-$k$ objective and the sum-objective respectively.
By definition of the top-$k$ norm, the following is true for every positive integer $k \leq n$, 
\begin{equation}\label{eq:bounds-norms}
    \cost_1(M) \leq \cost_k(M) \leq \cost_n(M) \leq n \cdot  \cost_1(M).
\end{equation}
Suppose that $\cost_1(M) \leq D_1 \cdot \cost_1(\Mopt_1)$ for some $D_1 \in \mathbb R$. We repeatedly apply \eqref{eq:bounds-norms} to get 
\begin{align*}
\cost_k(M) 
&\leq n \cdot \cost_1(M) \\
&\leq nD_1 \cdot \cost_1(\Mopt_1)  \\
&\leq nD_1 \cdot \cost_1(\Mopt_k)  \tag{by optimality of $\Mopt_1$}  \\
&\leq n D_1 \cdot \cost_k(\Mopt_k) 
\end{align*}
Suppose also $\cost_n(M) \leq D_n \cdot \cost_n(\Mopt_n)$ for some $D_n \in \mathbb R$. Similarly, we can show that 
\begin{align*}
\cost_k(M) 
&\leq \cost_n(M) \\
&\leq D_n \cdot \cost_n(\Mopt_n) \\
&\leq D_n \cdot \cost_n(\Mopt_1) \tag{by optimality of $\Mopt_n$} \\
&\leq nD_n \cdot \cost_1(\Mopt_1)  \\
&\leq nD_n \cdot \cost_1(\Mopt_k)  \tag{by optimality of $\Mopt_1$} \\
&\leq n D_n \cdot \cost_k(\Mopt_k) 
\end{align*}
It follows that for every positive integer $k \leq n$, 
\begin{equation*}
    \cost_k(M) \leq  \min ( nD_1,  nD_n  ) \cdot \cost_k(\Mopt_k),
\end{equation*}
which implies that $\fairness(M, d) \leq \min (nD_1, nD_n)$. Our desired result immediately follows.
\end{proof}

%% file: small-sum-large-max.tex
\begin{tikzpicture}[scale=0.75]
\tikzstyle{every node}=[draw=black,shape=circle, inner sep=2pt];

\pgfmathsetmacro{\k}{2}

\node (a1) at (\k +1,0) [fill=\agentcolor, color=\agentcolor, label=below:$a_1$] {};
\node (a2) at (2*\k + 2,0) [fill=\agentcolor,color=\agentcolor,label=below:$a_2$] {};
\node (a3) at (3*\k + 3,0) [fill=\agentcolor,color=\agentcolor,label=below:$a_3$] {};

\node (b1) at (1 + \k + 1,0) [fill=\itemcolor, color=\itemcolor, label=below:$b_1$] {};
\node (b2) at (1 + 2*\k + 2,0) [fill=\itemcolor, color=\itemcolor, label=below:$b_2$] {};
\node (b3) at (1 + 3*\k + 3,0) [fill=\itemcolor, color=\itemcolor, label=below:$b_3$] {};

\path[every node/.style={font=\footnotesize}]
(a1) edge node[yshift=0.2cm] {1} (b1) 
(b1) edge node[yshift=0.2cm] {$k$} (a2) 
(a2) edge node[yshift=0.2cm] {1} (b2)
(b2) edge node[yshift=0.2cm] {$k$} (a3) 
(a3) edge node[yshift=0.2cm] {1} (b3);
\end{tikzpicture}

%% file: small-max-large-sum.tex
\begin{tikzpicture}[scale=0.75]

    \pgfmathsetmacro{\sidelength}{3}
    \pgfmathsetmacro{\radiallength}{1}

    \tikzstyle{every node}=[shape=circle, inner sep=2pt]

    \node (a1) at (0,0) [fill=\agentcolor, color=\agentcolor, label=left:$a_1$] {};
    \node (b1) at (\sidelength, 0) [fill=\itemcolor, color=\itemcolor, label=right:$b_1$] {};
    \node (a3) at (\sidelength, \sidelength) [fill=\agentcolor, color=\agentcolor, label=right:$a_3$] {};
    \node (a2) at (0, \sidelength) [fill=\agentcolor, color=\agentcolor, label= left:$a_2$] {};

    \node (b3) at ([shift=(45:\radiallength)]a3)  [fill=\itemcolor, color=\itemcolor, label=above right:$b_3$] {};
    \node (b2) at ([shift=(135:\radiallength)]a2) [fill=\itemcolor, color=\itemcolor, label=above left:$b_2$] {};

    \path[every node/.style={font=\footnotesize, shape=rectangle, draw=none}]
        (a1) edge node[below] {$1$} (b1)
        (b1) edge node[right] {$D$} (a3)
        (a3) edge node[above] {$D$} (a2)
        (a2) edge node[left]  {$D$} (a1)
        
        (a3) edge node[xshift=-.12cm, yshift=.12cm] {0} (b3)
        (a2) edge node[xshift=.12cm, yshift=.12cm] {0} (b2);

\end{tikzpicture}

%% file: app-rep-match.tex
\begin{claim}\label{clm:property-of-logartihms}
If $f(x) = x^{\log_2 3}$ then   $2f(x) + f(y) \leq f(x +y)$ for all $0 < x \leq y$.
\end{claim}
\begin{proof}[Proof]
It suffices to show that the function $g(t) = (t +1)^c - 2t^c - 1$, where $c = \log_2 3$, is nonnegative for $t \in [0, 1]$. Indeed, when $t = x/y$, we have
\begin{align*}
0
&\leq y^c \cdot g(x/y)  \\
&= y^c \cdot \left ( (x/y + 1)^c - 2(x/y)^c - 1 \right )  \\
&= (x + y)^c - 2x^c - y^c \\
&= f(x + y) - 2f(x) - f(y).
\end{align*}
We now show that the function $g$ has the desired property. Taking derivatives, we have 
\begin{align*}
g'(t) 
&= c (t + 1)^{c-1} - 2 c t^{c-1} \\
g''(t)
&= c(c-1) (t + 1)^{c-2} - 2(c)(c-1) t^{c-2} \\
&= c(c-1) \left ( (t+1)^{c-2} - t^{c-2} \right).
\end{align*}
Note that for $t \in (0, 1)$, we have  $0 < \frac{t}{t+1} < 1$ and thus, $0 < (\frac{t}{t+1})^{2-c} < 1$. This in turn implies that $(t+1)^{c-2} - t^{c-2} < 0$ and thus, $g''(t) < 0$.  Hence, the function $g$ is strictly concave for every $t \in (0, 1)$. Since  $g(0) = 0$ and $g(1) = 0$, it follows that $g(t) \geq 0$ for every $t \in [0, 1]$, as desired. 
\end{proof}

\begin{remark}
Our choice of function $f$ in the proof of Theorem \ref{thm:mod-rep-match-ub-topk} is asymptotically optimal in the following sense. If $f$ is a non-decreasing function such that $f(1) = 1 $ and $2f(x) + f(y) \leq f(x + y)$ for all $0 < x \leq y$, then $f(n) \geq n^{\log_2 3}$ for infinitely many positive integers $n$.
\end{remark}

%% file: app-lower-bounds.tex
\begin{proof}[Proof of \Cref{thm:max-distortion-lower-bounds}]
Our approach is to construct a set of preferences for $n = 8$ agents and a family $\mathcal F$ of metrics consistent with those preferences. 
Regardless of how a matching $M$ assigns items to agents, there will always be a metric $d \in \mathcal F$ such that there  is an agent $a$ such that $d(a, M(a)) = 7$ which makes $\cost_1(M) \geq 7$, although the optimal matching $\Mopt_1$ for that metric will have cost $1$.

\paragraph{Agent Preferences.}
We define a metric space that makes it easy to determine the agent preferences. 
The metric space is defined in terms of an a weighted undirected tree with $8$ leaves\footnote{We define a \emph{leaf} to be a vertex with no children. An \emph{internal vertex} is a vertex with at least one child. Note that the root is not a leaf although it has degree $1$.}, such that the root has one child, and every other internal vertex has two children. Note that this graph has $8$ internal vertices. Agents will be located at the internal vertices and agents will be located at the leaves.
The weight of an edge $(u, v)$ for which vertex $u$ is the parent of vertex $v$ is $t + 1$, where $t$ is the number of edges along the path from vertex $u$ to the root. See \Cref{fig:lower-bound-pref} for a helpful diagram.

\begin{figure}[ht]
\centering
\input{lower-bound-pref}
\caption{The metric space that defines the preferences in our ordinal matching instance for $n = 8$. $b_8$ is the root of the tree. The black nodes are the items and the red (leaf) nodes are the agents.}\label{fig:lower-bound-pref}
\end{figure}
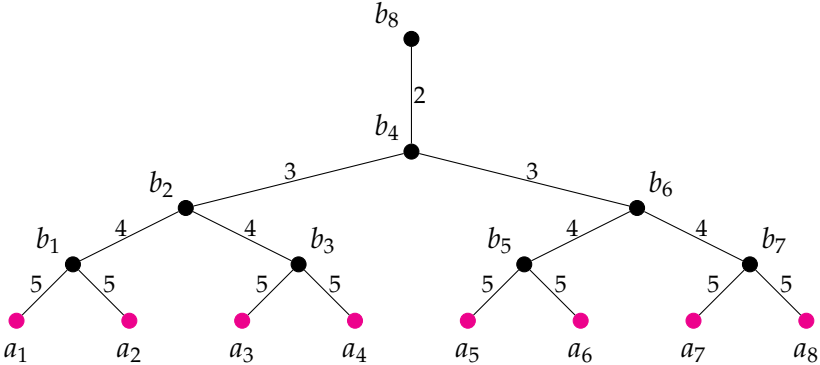

\paragraph{Consistent Metrics.}
We construct a family $\mathcal F$ of $8$ consistent metrics. Each agent $a_i$ will be associated with a unique metric $d_i$ in the family. The description of each metric is rather involved, but we provide the metric $d_1$ associated with the left-most leaf $a_1$ in \Cref{fig:lower-bound-hard}. By symmetry, one can readily construct the remaining  metrics. 

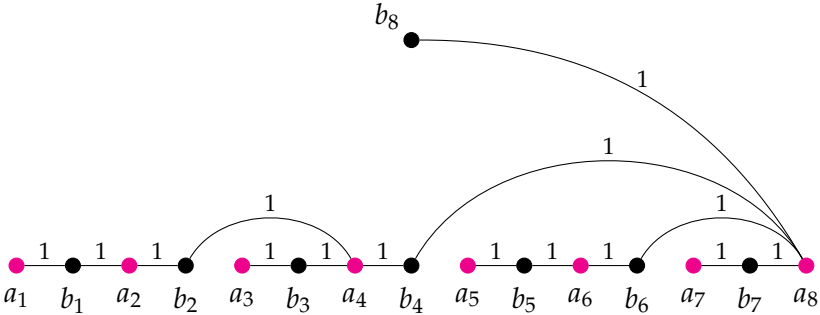
\begin{figure}[ht]
\centering
\input{lower-bound-hard}
\caption{The metric space that defines the preferences in the consistent metric $d_1$ for $n = 8$. $b_8$ is the root of the tree. The black nodes are the items and the red nodes are the agents.}\label{fig:lower-bound-hard}
\end{figure}

\paragraph{Establishing the Lower Bound.} Note that any perfect matching $M$ must assign some agent $a_i$ to the item $b_8$. 
It is straightforward to verify that for each metric $d_i$, the optimal matching has cost $1$ and the distance between agent $a_i$ and the root is $7$. For example, in the metric $d_1$ depicted in \Cref{fig:lower-bound-hard}, the matching that assigns each agent $a_i$ to item $b_i$ is optimal and we also have $d_1(a_1, b_8) = 7$. Therefore the matching $M$ will have distortion 7 in the metric space associated with $d_i$.
\end{proof}

%% file: lower-bound-pref.tex
\begin{tikzpicture}[scale=0.75]
\tikzstyle{every node}=[draw,shape=circle, inner sep=2pt];
\node (a1) at (2,0) [fill=\agentcolor,color=\agentcolor,label=below:$a_1$] {};
\node (a2) at (4,0) [fill=\agentcolor,color=\agentcolor,label=below:$a_2$] {};
\node (a3) at (6,0) [fill=\agentcolor,color=\agentcolor,label=below:$a_3$] {};
\node (a4) at (8,0) [fill=\agentcolor,color=\agentcolor,label=below:$a_4$] {};
\node (a5) at (10,0) [fill=\agentcolor,color=\agentcolor,label=below:$a_5$] {};
\node (a6) at (12,0) [fill=\agentcolor,color=\agentcolor,label=below:$a_6$] {};
\node (a7) at (14,0) [fill=\agentcolor,color=\agentcolor,label=below:$a_7$] {};
\node (a8) at (16,0) [fill=\agentcolor,color=\agentcolor,label=below:$a_8$] {};
\node (b1) at (3,1) [fill=\itemcolor,color=\itemcolor,label=above left:$b_1$] {};
\node (b2) at (5,2) [fill=\itemcolor,color=\itemcolor,label=above left:$b_2$] {};
\node (b3) at (7,1) [fill=\itemcolor,color=\itemcolor,label=above right:$b_3$] {};
\node (b4) at (9,3) [fill=\itemcolor,color=\itemcolor,label=above left:$b_4$] {};
\node (b5) at (11,1) [fill=\itemcolor,color=\itemcolor,label=above left:$b_5$] {};
\node (b6) at (13,2) [fill=\itemcolor,color=\itemcolor,label=above right:$b_6$] {};
\node (b7) at (15,1) [fill=\itemcolor,color=\itemcolor,label=above right:$b_7$] {};
\node (b8) at (9,5) [fill=\itemcolor,color=\itemcolor,label=above left:$b_8$] {};
\path[every node/.style={font=\footnotesize}]
(a1) edge node[xshift=-.11cm, yshift=.11cm] {5} (b1) 
(a2) edge node[xshift=.11cm, yshift=.11cm] {5} (b1)
(a3) edge node[xshift=-.11cm, yshift=.11cm] {5} (b3)
(a4) edge node[xshift=.11cm, yshift=.11cm] {5} (b3)
(a5) edge node[xshift=-.11cm, yshift=.11cm] {5} (b5)
(a6) edge node[xshift=.11cm, yshift=.11cm] {5} (b5)
(a7) edge node[xshift=-.11cm, yshift=.11cm] {5} (b7)
(a8) edge node[xshift=.11cm, yshift=.11cm] {5} (b7)
(b1) edge node[xshift=-.11cm, yshift=.11cm] {4} (b2)
(b3) edge node[xshift=.11cm, yshift=.11cm] {4} (b2)
(b5) edge node[xshift=-.11cm, yshift=.11cm] {4} (b6)
(b7) edge node[xshift=.11cm, yshift=.11cm] {4} (b6)
(b2) edge node[xshift=-.11cm, yshift=.11cm] {3} (b4)
(b6) edge node[xshift=.11cm, yshift=.11cm] {3} (b4)
(b4) edge node[xshift=.11cm] {2} (b8);
\end{tikzpicture}

%% file: lower-bound-hard.tex
\begin{tikzpicture}[scale=0.75]
\tikzstyle{every node}=[draw,shape=circle, inner sep=2pt];
\node (a1) at (2,0) [fill=\agentcolor,color=\agentcolor,label=below:$a_1$] {};
\node (a2) at (4,0) [fill=\agentcolor,color=\agentcolor,label=below:$a_2$] {};
\node (a3) at (6,0) [fill=\agentcolor,color=\agentcolor,label=below:$a_3$] {};
\node (a4) at (8,0) [fill=\agentcolor,color=\agentcolor,label=below:$a_4$] {};
\node (a5) at (10,0) [fill=\agentcolor,color=\agentcolor,label=below:$a_5$] {};
\node (a6) at (12,0) [fill=\agentcolor,color=\agentcolor,label=below:$a_6$] {};
\node (a7) at (14,0) [fill=\agentcolor,color=\agentcolor,label=below:$a_7$] {};
\node (a8) at (16,0) [fill=\agentcolor,color=\agentcolor,label=below:$a_8$] {};
\node (b1) at (3,0) [fill=\itemcolor,color=\itemcolor,label=below:$b_1$] {};
\node (b2) at (5,0) [fill=\itemcolor,color=\itemcolor,label=below:$b_2$] {};
\node (b3) at (7,0) [fill=\itemcolor,color=\itemcolor,label=below:$b_3$] {};
\node (b4) at (9,0) [fill=\itemcolor,color=\itemcolor,label=below:$b_4$] {};
\node (b5) at (11,0) [fill=\itemcolor,color=\itemcolor,label=below:$b_5$] {};
\node (b6) at (13,0) [fill=\itemcolor,color=\itemcolor,label=below:$b_6$] {};
\node (b7) at (15,0) [fill=\itemcolor,color=\itemcolor,label=below:$b_7$] {};
\node (b8) at (9,4) [fill=\itemcolor,color=\itemcolor,label=above left:$b_8$] {};
\path[every node/.style={font=\footnotesize}]
(a1) edge node[yshift=.2cm] {1} (b1) 
(a2) edge node[yshift=.2cm] {1} (b1)
(a2) edge node[yshift=.2cm] {1} (b2)
(a3) edge node[yshift=.2cm] {1} (b3)

(a4) edge node[yshift=.2cm] {1} (b3)
(a4) edge[bend right=60] node[yshift=.2cm] {1} (b2)
(a4) edge node[yshift=.2cm] {1} (b4)

(a5) edge node[yshift=.2cm] {1} (b5)

(a6) edge node[ yshift=.2cm] {1} (b5)
(a6) edge node[yshift=.2cm] {1} (b6)

(a7) edge node[yshift=.2cm] {1} (b7)

(a8) edge[bend right] node[yshift=.2cm] {1} (b8)
(a8) edge[bend right=60] node[yshift=.2cm] {1} (b4)
(a8) edge[bend right=60] node[yshift=.2cm] {1} (b6)
(a8) edge node[yshift=.2cm] {1} (b7);

\end{tikzpicture}

%% file: app-other-mechanisms.tex
We now briefly discuss how  other deterministic mechanisms fare under the max-distortion lens. 
We note that the sum-distortion of these mechanisms has already been studied. It turns out that many of the examples used to show lower bounds for the sum-distortion also provide similar lower bounds for the max-distortion  because most of the mechanism's cost is incurred by a single agent. For completeness, we provide full descriptions below.

We begin with the Serial Dictatorship, a well-known truthful mechanism for assignment problems. In a fixed an order over the agents, it assigns each agent her favorite item that is not yet assigned. The sum-distortion of this mechanism was shown by \cite{CFRF+16} to be exponential. We show that this is also the case for the max-distortion. 
\begin{proposition}\label{prop:serial-dictator-lb}
The Serial Dictatorship has worst-case max-distortion $\Omega(2^{n})$.
\end{proposition}
\begin{proof}
We revisit the instance used by \cite{CFRF+16} to study the sum-distortion of the Serial Dictatorship. This instance consists of $n$ agents and items configured along the real line, parametrized by a positive constant $\epsilon < 1$. Agents are located at the points 
$$
\{ 1 \} \cup  \{ 2^{t} : t \in \mathbb Z^+, t \leq n-1 \}.
$$
There is one agent at each of these $t$ locations. Agents are assigned items in order of their distance to $0$ (the closest goes first). Items are located at the points 
$$
\{ - \epsilon \}  \cup \{ 2^{t} : t \in \mathbb Z^+, t \leq n-1 \}.
$$ there is one item at each of these $t$ locations.

It is straightforward to verify that the optimal matching $\Mopt_1$ for the max-objective satisfies $\cost_1(\Mopt_1) = 1 + \epsilon$. Note however, that the Serial Dictatorship will assign the agent at $2^{n-1}$ to the item at $-\epsilon$ due to a domino effect, resulting in a matching $M$ with $\cost_1(M) = 2^{n-1} + \epsilon$. This implies the stated distortion bound. 
\end{proof}
\begin{remark}
The Serial Dictatorship belongs to a class of truthful algorithms that \cite{Anari:2023aa} call \emph{serializable}.  
It is not hard to show that this class similarly has exponential max-distortion.
\end{remark}

The Boston mechanism \cite{AS03} is another well-known deterministic mechanism and is commonly used in the school choice literature. The mechanism takes place over a series of rounds. In the first round, each agent proposes to her favorite item. Then, each item that receives a proposal is irrevocably matched to one its proposers and rejects all others. In the second round, the agents that failed to get matched then propose to their second favorite items, and the mechanism continues in the same way until all agents are matched. For simplicity, we will say that choices are made according to a fixed priority order, but this can be extended to the other settings (e.g., arbitrary tie-breaking). The sum-distortion of this mechanism was shown by \cite{Anari:2023aa} to be essentially exponential. We show that this remains the case for the max-distortion.

\begin{proposition}\label{prop:boston-mechanism-lb}
The Boston mechanism has worst-case max-distortion $\Omega(2^{\sqrt{n}})$.
\end{proposition}
\begin{proof}
We revisit the instance used by \cite{Anari:2023aa} to study the sum-distortion of the Boston mechanism. The agent and items locations are the similar to those in the proof of \Cref{prop:serial-dictator-lb}, but with duplication. This time, there are $n = 1 + k(k-1)/2$ agents and items. For each positive integer $t \leq k -1 $ there are $t$ agents and $t$ items at the point $2^t$. There is one agent at 1, and one item at $-\epsilon$. Co-located agents have identical preferences. Agents have priority in order of their distance to 0.

Like the Serial Dictatorship, the Boston mechanism suffers from a domino effect.
In the first round, for every $t \geq 2$, the agents at $2^t$ all propose to the same item at $2^t$, and that item is assigned to one of these agents. 
Meanwhile, both the agents at $1$ and $2$ propose to the item at $2$, and the item will be assigned to the agent at $1$. In the second round, for every $t \geq 3$ the agents at $2^t$, the unmatched agents at $2^t$ all propose to the same item at $2^t$, and that item is assigned to one of those agents. Meanwhile, both the unmatched agents at $2$ and $4$ propose to the item at $4$, and the item will be assigned to the agent at $2$. A similar situation occurs on subsequent steps.

Again, it is straightforward to verify that the optimal matching $\Mopt_1$ for the max-objective satisfies $\cost_1(\Mopt_1) = 1 + \epsilon$. It is also not hard to see that the mechanism returns a matching $M$ with $\cost_1(M) = 2^{k-1} + \epsilon $. Since $n = \Theta(k^2)$, the mechanism attains our stated distortion bound.
\end{proof}

To summarize, we present all known results on the distortion and fairness ratio of deterministic mechanisms in \Cref{tab:distortion-det-mechanisms}.
\begin{table}[!htb]
\centering

\renewcommand{\arraystretch}{1.5} 
\setlength{\tabcolsep}{10pt}  

\begin{tabular}{llll}
\toprule
\textbf{Mechanisms} & \textbf{Max} & \textbf{Sum} &\textbf{Fairness Ratio} \\
\midrule
Serial Dictatorship & $\Omega(2^{n})$ & $\Omega(2^{n})$ \cite{Anari:2023aa}& $\Omega(2^{n})$  \\
Boston   & $\Omega(2^{\sqrt{n}})$&  $\Omega(2^{\sqrt{n}})$ \cite{Anari:2023aa} & $\Omega(2^{\sqrt{n}})$   \\
\Cref{alg:rep-match} (worst-case operation)  & $\Omega(n); O(n^{1.58}) $ &  $\Theta(n^2) $ &  $O(n^2)$ \\
\bottomrule
\end{tabular}

\caption{Distortion and Fairness Ratio of Deterministic Mechanisms.}
\label{tab:distortion-det-mechanisms}
\end{table}